\documentclass{article}
\usepackage{amsmath}
\usepackage{amssymb}
\usepackage{amsthm}
\usepackage{graphicx}
\usepackage{comment}
\usepackage{enumerate}
\usepackage{tabularx}
\usepackage[noend]{algpseudocode}
\usepackage{algorithm}
\usepackage{mathtools}
\usepackage{cite}
\usepackage{url}

\newtheorem{theorem}{Theorem}

\newtheorem{lemma}[theorem]{Lemma}

\newtheorem{definition}[theorem]{Definition}

\newcommand{\shortcut}[1]{}

\newcommand{\say}[1]{} 
\newcommand{\opt}{\mathrm{OPT}}

\title{Network Design Problems with Bounded Distances via Shallow-Light Steiner Trees}
\author{
Markus Chimani\\
Theoretical Computer Science, Osnabr\"uck University, Germany\\
\url{markus.chimani@uni-osnabrueck.de}\\
\and
Joachim Spoerhase\\
Institute of Computer Science, University of W\"urzburg, Germany\\
\url{joachim.spoerhase@uni-wuerzburg.de}
}

\begin{document}


\maketitle

\begin{abstract}
  In a directed graph $G$ with non-correlated edge lengths and costs, the \emph{network 
design problem with bounded distances} asks for a cost-minimal spanning subgraph subject 
to a length bound for all node pairs.  We give a bi-criteria 
$(2+\varepsilon,O(n^{0.5+\varepsilon}))$-approximation for this problem. This improves on 
the currently best known linear approximation bound, at the cost of violating the 
distance bound by a factor of at most~$2+\varepsilon$.

  In the course of proving this result, the related problem of \emph{directed shallow-light 
Steiner trees} arises as a subproblem. In the context of directed 
graphs, approximations to this problem have been elusive. We present the first non-trivial 
result by proposing a $(1+\varepsilon,O(|R|^{\varepsilon}))$-ap\-proxi\-ma\-tion, where $R$ are the terminals.
  
  Finally, we show how to apply our results to obtain an $(\alpha+\varepsilon,O(n^{0.5+\varepsilon}))$-approximation for \emph{light-weight directed 
$\alpha$-spanners}. For this, no non-trivial approximation algorithm has 
been known before.  All running times depends on $n$ and
  $\varepsilon$ and are polynomial in $n$ for any fixed
  $\varepsilon>0$.
\end{abstract}

\clearpage

\section{Introduction}

We consider the following network design problem introduced by Dodis and
Khanna \cite{dodis-khanna99:network-design-bounded-dist}: 
\begin{definition}[Directed Network Design with Bounded Distances]\label{def:dndbd}
Given a directed graph $G=(V,E)$, an edge cost function
$c\colon E\to\mathbb{N}$, an edge length function $\ell\colon
E\to\mathbb{N}$, and a length bound $L\in\mathbb{N}$.  We ask for a 
spanning subgraph $H$ of $G$ of minimum cost (with
respect to $c$) such that for each node pair $u,v$ the distance
in $H$ (with respect to $\ell$) is at most~$L$. 
\end{definition}

Generally, for a given graph $G=(V,E)$, we let $n:=|V|$ and $m:=|E|$; $\bar\ell_H(u,v)$ denotes the lengths of the 
shortest $u$-$v$ path in $H\subseteq G$ with respect to $\ell$.
For uniform edge costs and lengths, Dodis and Khanna \cite{dodis-khanna99:network-design-bounded-dist}
devise an $O(\log n\log L)$-approximation. For non-uniform edge costs, they show 
$\Omega(2^{\log^{1-\varepsilon}n})$-hardness of approximation, and propose
an $O(n\log L)$-approximation under the restriction that the edge lengths are polynomially bounded.
Up to now, no improved algorithm is known.

In this paper (Section~\ref{sec:network-design-bounded}), we give an algorithm 
for this problem, without any of the above restrictions and without ratio-dependency on $L$, achieving essentially 
a performance ratio $O(\sqrt n)$ while 
violating the distance bound~$L$ by a factor of at most~$2+\varepsilon$. 
\begin{theorem}
There is a bi-criteria $(2+\varepsilon,O(n^{1/2+\varepsilon}))$-approximation for the above directed network design problem with bounded distances.
\end{theorem}

As a starting point, our algorithm uses a two-stage approach originally proposed by Feldman et al.\
\cite{feldman-etal12:dsf} for directed Steiner forest, which has later
been reused for directed spanners
\cite{bhattacharyya-etal:tcspanners,dinitz11:directed-spanners,berman-etal:apx-spanners-dsf}.  We
divide the considered node pairs into \emph{thin} and \emph{thick}
pairs. We settle the former by LP-rounding, as we have to cover
certain cuts w.r.t.\ shortest paths. For the latter, we sample nodes
and construct short in- and out-trees for each of them. This latter
part is a main technical challenge: In contrast to the case of sparse
spanners, we cannot simply use shortest-path trees, as they could have
arbitrarily high costs.
To solve this issue, we turn our attention to a second problem, which is also of independent interest:
\begin{definition}[Directed Shallow-Light Steiner Trees]
Given a directed graph $G=(V,E)$, an edge cost function $c\colon E\rightarrow\mathbb{N}$, an edge length function $\ell\colon E\rightarrow\mathbb{N}$, a distinguished root node $r\in V$, and
a set $R\subseteq V$ of terminals with distance bounds $d\colon R\to\mathbb{N}$. We ask for an
$r$-rooted subtree $T$ of $G$ of minimum cost (with respect to~$c$)
such that for any terminal $v\in R$ the distance $\bar\ell_T(r,v)$ in $T$ (with respect to~$\ell$) is at most~$d(v)$.
\end{definition}

Kortsarz and Peleg~\cite{kortsarz-peleg:slt} gave an $O(|R|^{\varepsilon})$-approximation for undirected graphs with uniform edge lengths and uniform distance bounds.
The directed problem with non-uniform edge costs has formerly been considered in \cite{naor97:wrong-approx-slst}, where a bi-criteria $(2,O(\log n))$-approximation for directed shallow-light \emph{spanning} trees (that is, $R=V$) was proposed.
Unfortunately, the proof is intrinsically flawed\footnote{Verified by personal communication with J.\ Naor.}, and there has not been any progress on the problem since.  We propose the first non-trivial result for the general directed problem (cf.\ Section~\ref{sec:direct-shall-light}). In fact, at the cost of violating the length bounds by a factor of at most $(1+\varepsilon)$, we obtain the same approximation ratio as~\cite{kortsarz-peleg:slt}, but for directed graphs and without the restrictions to uniform lengths and costs:
\begin{theorem}\label{thm:slst}
There is a bi-criteria  $(1+\varepsilon,|R|^{\varepsilon})$-approximation for directed shallow-light Steiner trees.
\end{theorem}

Finally (Section~\ref{sec:lwds}), we give
a further application of our shallow-light Steiner tree result: 
\begin{definition}[Light-Weight Directed $\alpha$-Spanners]\label{def:alphaspan}
Given a directed graph $G=(V,E)$, an edge cost function
$c\colon E\to\mathbb{N}$, an edge length function $\ell\colon
E\to\mathbb{N}$, and a stretch factor $\alpha\geq 1$. 
We ask for a 
spanning subgraph $H$ of $G$ of minimum cost (with respect to~$c$) 
such that for each node pair $u,v$ the distance
$\bar\ell_H(u,v)$ in $H$ (with respect to~$\ell$) is at most $\alpha\cdot\bar\ell_G(u,v)$, i.e., $\alpha$
times their distance in $G$.
\end{definition}
As of now, this problem has only been successfully tackled for undirected graphs~\cite{peleg-schaeffer:graph-spanners,althoefer-etal:spanners}. Its directed variant remained
an interesting open problem~\cite{dinitz11:directed-spanners}\footnote{As mentioned in the corresponding slides, available online.}.
We give the first non-trivial result:
\begin{theorem}
There is a bi-criteria $(\alpha+\varepsilon,O(n^{1/2+\varepsilon}))$-approximation for light-weight directed $\alpha$-spanners.
\end{theorem}

\section{Network Design with Bounded Pairwise Distance}\label{sec:network-design-bounded}

We build our solution network as the union of subgraphs. We say such a subgraph \emph{settles} a node pair $(u,v)$, if it includes a path connecting $u$ to $v$ complying with the
distance bound.
As sketched above, the overall scheme of our approximation algorithm is to classify node pairs into two categories.
Let $(u,v)\in V\times V$ be any node pair, and $\mathcal{P}^L_{uv}$ the set of all
$u$-$v$ paths of length at most~$L$. We denote with
$V_{uv}:=\bigcup_{P\in\mathcal{P}^L_{uv}}V(P)$ and $E_{uv}:=\bigcup_{P\in\mathcal{P}^L_{uv}}E(P)$ 
the nodes and edges, respectively, contained in any such path. 
The node pair $(u,v)$ is called \emph{thin} if $|V_{uv}|\leq \sqrt{n}$ and
\emph{thick} otherwise.
We settle node pairs based on this classification. However, we will never explicitly compute any $\mathcal{P}^L_{uv}$, $V_{uv}$, $E_{uv}$ nor any node-pair classifications.
They are only of interest for the approximation proof.
We note that the concept of this classification is lifted from Feldman et al.\ \cite{feldman-etal12:dsf}.
The handling of the thin pairs follows the idea of anti-spanners by Berman et al.\ \cite{berman-etal:apx-spanners-dsf}, as it can be made to work in our context, see below. Successfully tackling the thick pairs, however, is a technical challenge and requires our result on shallow-light trees (see Section~\ref{sec:direct-shall-light}).
Let $\opt$ denote the value of the optimum solution to the full problem.

\subsection{Thin Pairs}

\paragraph{Path-Based LP.}
We consider the following path-based
LP relaxation of the problem, requiring an exponential number of variables. Let $\mathcal{P}^L:=\bigcup_{(u,v)\in V\times V} \mathcal{P}^L_{uv}$.

\begin{align}
\begin{aligned}
  \min \sum\nolimits_{e\in E} c_ex_e, & \quad \textnormal{s.t.} \\  
  \sum\nolimits_{P\in\mathcal{P}^L_{u,v}} f_P & \geq 1  && \forall (u,v)\in V\times V\\
  \sum\nolimits_{P\in\mathcal{P}^L_{u,v},P\ni e}f_P & \leq x_e && \forall e\in E, (u,v)\in V\times V\\
  x_e \geq 0,\quad f_P & \geq 0 && \forall e\in E,\quad \forall P\in\mathcal{P}^L
\end{aligned}\label{primalLP}\\
\intertext{Its dual can be written as:}
\begin{aligned}
  \max \sum\nolimits_{(u,v)\in V\times V}\alpha_{uv}, & \quad \textnormal{s.t.} \\
  \sum\nolimits_{e\in P}\beta^e_{uv} & \geq \alpha_{uv} && \forall(u,v)\in V\times V, P\in\mathcal{P}^L_{uv}\\
  \sum_{(u,v)\in V\times V}\beta^e_{uv} & \leq c_e && \forall e\in E\\
  \beta^e_{uv}\geq 0,\quad \alpha_{uv} & \geq 0 && \forall e\in E, (u,v)\in V\times V
\end{aligned}\label{dualLP}
\end{align}

LP~(\ref{primalLP}) has an exponential number of variables. Below, we
argue that we can get a PTAS for this LP by an approach analogous to
the one proposed in \cite{dinitz11:directed-spanners}. Let
$\varepsilon>0$.  We first consider the dual LP~(\ref{dualLP}).  This
LP has a polynomial number of variables but an exponential number of
constraints. We use the ellipsoid method to get an approximate
solution to it.  The separation oracle works as follows. (We do not
consider the constraints $\sum_{(u,v)\in V\times V}\beta^e_{uv} \leq
c_e$ since there are only polynomially many of these.)  For each fixed
$(u,v)\in V\times V$, we consider variables the $\beta^e_{uv}$ as edge
weights.  Thus, determining whether a constraint is violated for some
$P\in P^L_{uv}$ amounts to checking whether $\alpha_{uv}$ is at most
the weight of a lightest $u$--$v$ path (under weights $\beta^e_{uv}$)
whose length (under edge lengths $\ell$) is bounded by $L$.  Already
this necessary subproblem (\emph{length-bounded shortest path}) is
NP-hard.  However, Hassin~\cite{hassin92:fptas-restr-shortest-paths},
later sped up by Ergun et al.~\cite{ergun-etal:fptaspath}, describes
an FPTAS.  Assume we run the ellipsoid algorithm by using this
approximate separation oracle with error parameter $\varepsilon$.
Then, we end up with an optimum solution to the \emph{restricted} dual
LP, which has only constraints for paths $P\in\mathcal{P}^L$ that we
included when running the ellipsoid algorithm. Since we used an FPTAS
for the separation oracle, the constraints that we did not include can
be violated by a factor at most $1-\varepsilon$. That is, we have
$\sum\nolimits_{e\in P}\beta^e_{uv} \geq (1-\varepsilon)\alpha_{uv}$
for all paths $P\in\mathcal{P}^L$ that we did not include.  Hence, if
we set $\alpha_{uv}'=(1-\varepsilon)\alpha_{uv}$ we obtain a feasible
solution to the original dual LP that is $(1-\varepsilon)$-approximate
with respect to the optimum solution of the restricted dual.  Now
suppose that we solve the \emph{restricted} primal LP where we only
include the (polynomially many) variables that correspond to
constraints of the restricted dual.  Then the optimum solution to this
LP is at most $1/(1-\varepsilon)$ times larger than the optimum
solution to the original dual (and hence the original primal) since
the restricted dual LP is the dual to the restricted primal LP and
since the original dual is $(1-\epsilon)$-approximate to the
restricted dual.

\newcommand{\Hthin}{H_1}
\paragraph{Randomized LP Rounding.}
We describe an algorithm that computes a subgraph $\Hthin\subseteq G$ where
the distance $\bar \ell_{\Hthin}(u,v)$ is at most $L$ for every thin pair $(u,v)$.
The algorithm first solves the
above LP within a ratio of $1+\varepsilon$.  Then each edge $e$ is sampled with probability
$\min(\gamma\cdot x_e,1)$ where $\gamma:=\sqrt{n}\cdot\log n$.  The cost of $\Hthin$ is $O(\gamma(1+\varepsilon)\opt)$.  We have to show
that this algorithm creates a feasible solution with high probability.

\begin{definition}
  Let $(u,v)$ be a thin pair, $C\subseteq E$ a set of edges, and $G_C:=(V,E\setminus C)$.
  We say $C$ is a \emph{$u$-$v$-stretching cut} if $\bar\ell_{G_{C'}}(u,v)\leq L$ for all $C'\subset C$ but $\bar\ell_{G_C}(u,v)>L$.
\end{definition}

\begin{lemma}\label{lem:settling-by-cuts}
  Let $H=(V,E')$ be a subgraph of $G$ and $(u,v)$ a thin pair.  $H$
  settles $(u,v)$ if and only if each $u$-$v$-stretching cut contains
  at least one edge of $E'$.
\end{lemma}
\begin{proof}
  If there is a $u$-$v$-stretching cut $C$ that contains no edge of
  $E'$ then $E'\subseteq E\setminus C$ and hence
  $\bar\ell_{H}(u,v)\geq\ell_{G_C}>L$.  Conversely, if $H$ does not
  settle $(u,v)$ then $\bar\ell_{H}(u,v)>L$ and hence $E\setminus E'$
  would contain a $u$-$v$-stretching cut $C$, which clearly has no
  edge of $E'$.
\end{proof}

\begin{lemma}\label{lem:number-stretching-cuts}
  For each thin pair $(u,v)$ the number of $u$-$v$-stretching cuts is
  at most $\sqrt{n}^{\sqrt{n}}$.
\end{lemma}
\begin{proof}
  Consider some $u$-$v$-stretching cut $C$ and let $T$ be a shortest
  path tree in the graph $H_C:=(V_{uv},E_{uv}\setminus C)$ rooted at $u$. Let
  $\bar\ell_T(w)$ denote the distance from $u$ to $w$ in~$T$.  If there is no
  $u$-$w$ path in $H_C$ then $\bar\ell_T(w):=\infty$.  We show that $C=\{wx\in
  E_{uv}\mid \bar\ell_T(w)+\ell(wx)<\bar\ell_T(x)\}$, which implies that $C$ is
  uniquely determined by $T$.

  Consider an edge $wx\in E_{uv}$ such that
  $\bar\ell_T(w)+\ell(wx)<\bar\ell_T(x)$.  Then $wx\in C$ because $T$ is a
  shortest path tree in $H_C$.

  Now, let $wx\in C$.  Because $C':=C\setminus\{wx\}$ is not a $u$-$v$ stretching
  cut there is a $u$-$v$ path in $H_{C'}:=(V_{uv},E_{uv}\setminus C')$ of length at
  most~$L$.  This path must use the edge $wx$ and has length
  $\bar\ell_T(w)+\ell(wx)+\bar\ell_{H_C}(x,v)$.  Since $H_C$ has no $u$-$v$
  path of length at most $L$ we can conclude that
  $\bar\ell_{H_C}(u,x)+\bar\ell_{H_C}(x,v)>L$ and therefore
  $\bar\ell_T(w)+\ell(wx)<\bar\ell_{H_C}(u,x)=\bar\ell_T(x)$.

  Hence the $u$-$v$-stretching cut $C$ is uniquely determined
  by the tree $T$.  We now count the number of rooted trees in $H_C$.
  For every node in such an out-tree there are $\sqrt{n}$ possibilities
  to choose its parent node.  Hence the total number of rooted trees
  and therefore the number of $u$-$v$-stretching cuts can be upper
  bounded by $\sqrt{n}^{\sqrt{n}}$.
\end{proof}

\begin{lemma}\label{lem:thinsettled}
  The above algorithm settles each thin pair with high probability.
\end{lemma}
\begin{proof}
  By Lemma~\ref{lem:settling-by-cuts}, is suffices to show that for
  every thin pair $(u,v)$ and every $u$-$v$ stretching cut $C$ there
  is an edge from $\Hthin$ in $C$ with high probability.

  For every such cut $C$ the LP value $\sum_{e\in C}x_e$ must be at
  least~1.  This holds because every $u$-$v$ path in
  $\mathcal{P}_{u,v}^L$ must contain at least one edge of $C$, since
  the total flow sent along these paths is at least~1 and since
  $\sum_{e\in C}x_e$ is an upper bound on this total flow because of
  the contraints $\sum_{P\in\mathcal{P}^L_{u,v},P\ni e}f_P \leq x_e$
  in the LP.  If $\gamma\cdot x_e\geq 1$ for some $e\in C$ then $e\in
  E(\Hthin)$.  Otherwise, the probability that none of the edges in
  $C$ is sampled is at most
  \begin{displaymath}
    \prod\nolimits_{e\in C}(1-\gamma x_e)\leq \prod\nolimits_{e\in C}e^{-\gamma x_e}=e^{-\sqrt{n}\cdot \log n\sum_{e\in C} x_e}\leq n^{-\sqrt{n}}\,.
  \end{displaymath}

  By Lemma~\ref{lem:number-stretching-cuts}, the total number of
  stretching cuts is at most $n^2\sqrt{n}^{\sqrt{n}}$.  Hence the
  probability that at least one stretching cut contains no edge of $\Hthin$
  is at most $\sqrt{n}^{-\Omega(\sqrt{n})}$.
\end{proof}

\subsection{Thick Pairs and Overall Algorithm}

We now describe an algorithm to settle all thick pairs.  The algorithm
samples a set of $\delta=3\sqrt{n}\log n$ many nodes of $G$.  For
each node $u$ in this set, the algorithm determines a $u$-rooted shallow-light
Steiner tree $T_u$ by means of the algorithm described in
Section~\ref{sec:direct-shall-light} and summarized in
Theorem~\ref{thm:slst}.  As input for this algorithm we use the graph
$G$, the edge costs $c$ and the edge lengths $\ell$ as in the instance of the
network design problem;
 the root is the node~$u$ and the set $R$ of
terminals are all $V\setminus\{u\}$; we use $L$ as the distance bound for each node.
Similarly, the algorithm
computes an in-tree rooted at $u$ such that for each node the distance
to $u$ is at most $L$.  This can be accomplished by computing a shallow-light Steiner tree $T'$ in the graph $G'$ arising from $G$ by reversing all edges and then reversing the edges of $T'$.
The
output $H_2$ of the process is the union of all these spanning trees.

Our overall algorithm then returns $H_1\cup H_2$, the union of the solution for
the thin and the thick pairs, respectively.
We are now ready to prove the following theorem:
\setcounter{theorem}{0}
\begin{theorem}[Revisited]
  The above algorithm is a bi-criteria
  $(2+\varepsilon,O(n^{1/2+\varepsilon}))$-approximation algorithm for
  the directed network design problem with bounded distances (cf.\
  Definition~\ref{def:dndbd}). The running time depends on $n$ and
  $\varepsilon$ and is polynomial in $n$ for any fixed
  $\varepsilon>0$.
\end{theorem}
\begin{proof}
  We first show that the algorithm outputs a feasible solution with
  high probability. In the light of Lemma~\ref{lem:thinsettled}, it remains to show that all thick pairs
  are settled with high probability.  A thick pair $(u,v)$
  is settled if the above algorithm samples a node $r$ from the set~$V_{uv}$.  
  In this case, the inclusion of the $r$-rooted in-tree and
  the $r$-rooted out-tree guarantees the existence of a $u$-$v$ path
  of length at most $2(1+\varepsilon)L$: we travel from $u$ to $r$ and then from
  $r$ to $v$.  Since for any thick pair its set $V_{uv}$ contains at
  least $\sqrt{n}$ many nodes, the probability that none of the $\delta$
  many sampled nodes are from $V_{uv}$ can be bounded by
  \begin{displaymath}
    \left(1-\frac{1}{\sqrt{n}}\right)^\delta\leq e^{-3\log n}=\frac{1}{n^3}\,.
  \end{displaymath}
  Since there are at most $n^2$ thick pairs the claim follows.

  We now analyze the cost of the algorithm.  The cost of the procedure
  for settling thin pairs is $\gamma(1+\varepsilon)\opt$
   since every edge is sampled
  with probability at most $\gamma$ times higher than its LP value.
  Now observe that every tree constructed in the procedure for thick
  pairs has cost at most $O(n^{\varepsilon})\opt$.  This follows from the
  fact that the optimum solution to the network design problem
  ensures the existence of a feasible solution to the problem of finding the rooted
  subtrees, and that the algorithm from
  Section~\ref{sec:direct-shall-light} is an
  $O(n^\varepsilon)$-approximation algorithm.  Since the number of such
  trees constructed by the algorithm is $O(\delta)$ the ratio of the algorithm is bounded by $O(\delta
  n^{\varepsilon}+\gamma)=O(n^{1/2+\varepsilon})$.
\end{proof}

\section{Directed Shallow-Light Steiner Trees}\label{sec:direct-shall-light}

Let $T$ be a rooted out-tree, i.e., its edges are directed from the root 
towards the leaves. A \emph{branch node} is a node with out-degree larger 
than~$1$; as a special case, we always consider the root node to be a branch 
node. We say $T$ is an $i$-level tree if no path from the root to any leaf 
contains more than $i$ branch nodes.

Let $T\subseteq G$ be any out-tree, subgraph of a complete digraph $G$, with an arbitrary number of levels. 
Clearly, we can find a related out-tree with the same root and leaves requiring at most $i$ levels, for any given $i$. If the edges
have metric weights, a very general result by Helvig et al.~\cite{helvig-etal:groupSteiner}
relates the weights of these two trees:
\begin{lemma}[Helvig et al.~\cite{helvig-etal:groupSteiner}]\label{withoutlevels}
 Let $T$ be a rooted subtree of weight $c(T)$ with $k$ leaves in a metrically-weighted complete digraph, and $T_i$ the cheapest subtree with the same root and leaves and at most $i$ levels.
 We have $c(T_i)\leq 2i (k/2)^{1/i} c(T)$.
\end{lemma}

A typical application of this lemma is the following: Assuming metric edge weights, any digraph can be considered complete by adding artificial edges corresponding to paths in $G$. Consider any optimization problem whose solution is a tree. We can establish an approximation algorithm for it by first finding an approximation for the best $p$-level solution, for some~$p$. We can then apply the lemma to obtain an approximation ratio to the original non-level-restricted problem. In our application, we have non-correlated edge costs and lengths. However, in order to apply the lemma, it suffices to observe that if there is a node pair $(u,w)$ without any edge $uw$ of length at most
$\ell(uv)+\ell(vw)$, for any node $v$, we could (conceptually) insert an edge with this length and cost $c(uv)+c(vw)$ representing this $u$-$v$-$w$ path. Observe that this would, in general, result in multiple edges connecting
the same node pair, with different length/cost combinations. We do not need to explicitly consider these additional edges. In our algorithm, we will directly identify the corresponding paths meeting at branch nodes. Furthermore, by adding edges of zero length and cost, we can in the following always assume that there is an optimum solution where all terminals appear as leaves.

\subsection{Algorithm}

As mentioned above, there is an FPTAS~\cite{hassin92:fptas-restr-shortest-paths,ergun-etal:fptaspath}
to solve the problem of finding the cheapest (with respect to edge costs~$c$) path from a node $u$ 
to a node $v$ of length at most~$D$ (with respect to edge length~$\ell$). We denote 
the result of this FPTAS by \textsc{MinCostPath}$(u,v,D)$.

Our algorithm employs a recursive greedy strategy, which has been originally
invented by Zelikovsky \cite{zelikovsky:acyclicSteiner}.  It has later
been applied by Kortsarz and Peleg \cite{kortsarz-peleg:slt} to 
undirected Shallow-Light Steiner Trees.  Specifically, they
give an $(2+\varepsilon,O(|R|^{\varepsilon}))$-approximation for
undirected graphs with uniform edge lengths and uniform distance
bounds.  Charikar et al. \cite{charikar99:approx-dst} reuse this
strategy for directed Steiner trees (without distance bounds) and
obtain an $O(|R|^{\varepsilon})$-approximation algorithm, devising
a particularly elegant analysis of recursive greedy.

Our algorithm uses five parameters, cf.~Algorithm~\ref{alg:algorithm}. The graph~$G$, costs $c$, and lengths~$\ell$ 
remain unchanged over all recursive calls to the procedure and are hence not explicitly included in these parameters.
The algorithm operates in \emph{levels} given by parameter $i\leq n$.
The higher the level, the better the approximation guarantee.
Parameters $r,R$, and $d$ denote the root, the terminal set, and the vector of
distance bounds, respectively.  Parameter $k\leq |R|$ specifies the
minimum number of terminals out of $R$, the resulting tree has to span (while meeting the distance bounds).
Setting $k=|R|$, the algorithm outputs a feasible directed shallow-light Steiner
tree.

Level $i=1$ of the algorithm works as follows.  For all terminals
$t\in R$, the algorithm computes an $r$-$t$ path $P_t$ by 
\textsc{MinCostPath}$(r,t,d(t))$.  Clearly, $P_t$
respects the length bound~$d(t)$.  The resulting tree consists of the union
of the $k$ cheapest (w.r.t.~$c$) of these paths.\footnote{As a side note, 
observe that one may be tempted to assume that some of these paths may coincide in the beginning, thus giving rise to a branch node
where the paths start to differ. We would hence, inadvertently, construct a tree with more than one level. We do not need to care about this issue: Firstly, 
in our cost computation (of the upper bound) we assume the worst case, i.e., that such common subpaths do not exist; if they would, the cost would 
only decrease, thus improving the approximative solution. Secondly, we can always (implicitly) consider the metric closure of $G$ (with multiedges
for different length-vs.-cost combinations); in this case we always find distinct paths.}

For $i>1$ we employ a greedy strategy to obtain a feasible solution $T$. 
Let the \emph{relative cost} of a tree $T'$ spanning $k'$
terminals be defined as $\varrho(T'):=c(T')/k'$.
Starting with empty~$T$, we iteratively
compute a subtree $T_\textnormal{best}$ of low relative cost
$\varrho(T_\textnormal{best})$, add it to $T$, remove the newly spanned
terminals from $R$, and adjust $k$ accordingly.

In order to compute $T_\textnormal{best}$, the algorithm exhaustively
tests all nodes $v$ and all values $k'\leq k$ to compute a cheap
tree $T'$ rooted at $v$ that spans at least $k'$ terminals.  (Note, that $k$ is adjusted by the algorithm.)  These
trees $T'$ are computed by applying the algorithm recursively but for
level $i-1$.  To obtain an $r$-rooted tree we connect $r$ to $v$ by a
path $P$.  This requires to adjust the distance bounds
accordingly in the above mentioned recursive calls.  An issue that
arises here is that the necessary properties of path $P$ are not clear a priori.
In general, we may not be able to use the shortest path (w.r.t.~$\ell$) 
as this might be too expensive (w.r.t.~$c$) to give a low relative cost.

To this end, we consider every possible path length up to $\ell(E)$, 
where the latter denotes the total length of all edges. This becomes tractable when we allow for 
a relative error of up to $(1+\varepsilon)$: we evaluate a geometrically increasing sequence of
length bounds $(1+\varepsilon)^j$, for non-negative integrals  $j$, and determine 
for each of these bounds the cheapest path $P_j$ respecting it.

\begin{algorithm}[tb]
  \caption{\label{alg:algorithm}Approximation of a directed shallow-light Steiner tree for $(G, c, \ell, r, R, d)$}
  \begin{algorithmic}[1]
  \Procedure{ShallowLight}{$i,r,R,d,k$}
    \If{no $k$ terminals in $R$ respect the distance bounds from $r$}
       \State \textbf{return} $\emptyset$
    \EndIf
    \If{$i=1$}
       \For{each terminal $t\in R$}
          \State $P_t\gets$\Call{MinCostPath}{$r,t,d(t)$}\label{linecompl1}
       \EndFor
       \State let $R'$ be the set of $k$ terminals with minimum $c(P_t)$
       \State \textbf{return} $\bigcup_{t\in R'}P_t$
    \EndIf
    \State $T\gets\emptyset$
    \While{$k>0$}\label{whileloop}
       \State $T_\textnormal{best}\gets\emptyset$
       \For{each $v\in V$ and each $k',1\leq k'\leq k$}\label{lineconsidervk}
          \For{$j=0,\dots,\lceil\log_{1+\varepsilon}\ell(E)\rceil$}
             \State $P_j\gets$\Call{MinCostPath}{$r,v,(1+\varepsilon)^j$}\label{searchbegin}
             \State $d'(u)\gets d(u)-\frac{\ell(P_j)}{1+\varepsilon}$ for each $u\in V$\label{linereduce}
             \State $T'\gets$\Call{ShallowLight}{$i-1,v,R,d',k'$}$\cup P_j$
             \If{$\varrho(T_\textnormal{best})>\varrho(T')$}
                $T_\textnormal{best}\gets T'$
             \EndIf\label{searchend}
          \EndFor
       \EndFor
       \State $T\gets T\cup T_\textnormal{best}$\label{addtreeline}
       \State $k\gets k-|R\cap V(T_\textnormal{best})|$
       \State $R\gets R-V(T_\textnormal{best})$
    \EndWhile
    \State \textbf{return} $T$
  \EndProcedure
  \end{algorithmic}
\end{algorithm}

\subsection{Analysis}

Let $\mathcal{G}:=(G,c,\ell,r,R,d)$ be a directed shallow-light 
Steiner tree problem instance as defined above. For the related problem of a 
\emph{$k$-terminal directed shallow light Steiner tree ($k$-DSLST)} we are given 
an instance $(\mathcal{G},k)$, $k\leq |R|$, and ask for the cheapest directed shallow 
light Steiner tree subject to any $k$-element subset of $R$. We observe that $k=|R|$ gives the original problem.
An \emph{$f(k)$-partial approximation} for $k$-DSLST is a procedure that finds a tree $T$ that is rooted at $r$, 
contains $1\leq k'\leq k$ terminals of $R$, and has relative cost $\varrho(T)\leq f(k)\cdot 
c(T^*)/k$. Here, $c(T^*)$ is the cost of an optimum solution to $k$-DSLST.

We will show later (cf.\ Lemma~\ref{partialapprox}) that the core of our algorithm in fact constitutes such a 
partial approximation. This allows us to adapt a lemma by Charikar et al.~\cite{charikar99:approx-dst} to obtain an approximation to 
the original problem, as summarized in the following lemma. While their result 
is dealing with Steiner trees and does hence not consider length restrictions, 
their proof is versatile enough to be carried out in an identical fashion for 
our following situation:
Let $\mathcal{P}(\mathcal{G},k)$ be a partial approximation routine. We 
construct an approximation algorithm $\mathcal{A}(\mathcal{G},k)$ as follows: 
First, $\mathcal{A}(\mathcal{G},k)$ calls $\mathcal{P}(\mathcal{G},k)$ which 
yields a tree $T'$ spanning some terminals $R'$. If $|R'|=k$, we are done. 
Otherwise, $\mathcal{A}(\mathcal{G},k)$ returns the union of $T'$ and the tree 
$T''$ resulting from $\mathcal{A}(\mathcal{G}'',k'')$ where $\mathcal{G}''$ is 
the problem instance with reduced terminal set $R\setminus R'$ and $k'':=k-|R'|$.

\begin{lemma}[Adaptation of Charikar et al.~\cite{charikar99:approx-dst}]\label{partialtofull}
 Given an $f(k)$-partial approximation $\mathcal{P}(\mathcal{G},k)$ and an 
algorithm $\mathcal{A}(\mathcal{G},k)$ as described above. If $f(x)/x$ is a 
decreasing function in $x$, then $\mathcal{A}$ is a $g(k)$-approximation, with 
$g(k)=\int_0^k(f(x)/x)dx$.
\end{lemma}

In the light of $\mathcal{P}$ and $\mathcal{A}$, the identification of 
$T_\textnormal{best}$ in Algorithm~\ref{alg:algorithm} corresponds to 
$\mathcal{P}$ while the outer while loop resembles $\mathcal{A}$.
It remains to show that our algorithm meets the criteria of an $f(k)$-partial 
approximation with $f(x)/x$ being a decreasing function. At its core, the proof 
strategy is similar to Charikar et al., but we have to carefully consider our 
length restrictions and violations within the recursion.

\begin{lemma}\label{partialapprox}
Consider \textsc{ShallowLight}$(i,r,R,d,k)$ (Alg.~\ref{alg:algorithm}), 
which iteratively computes~$T$. Let $\bar T:=T_\textnormal{best}$ be any tree 
incorporated in the current solution (line~\ref{addtreeline}). It violates the 
length bounds by a factor of at most $(1+\varepsilon)$. For $i\geq 2$, $\bar 
T$'s relative cost $\varrho(\bar T)$ is at most $(i-1)$ times the relative cost 
$\varrho^*:=\varrho^*_{\bar R,\bar k}$ of the optimum solution $T^*:=T^*_{\bar 
R,\bar k}$ to $\bar k$-DSLST with $i$ levels, where $\bar R$ and $\bar k$ are the 
values for $R$ and $k$ currently used by the algorithm, respectively.
\end{lemma}
\begin{proof}
Observe that, for $i>1$, $\bar T$ consists of an $r$-$v$ path $\bar P$ and a tree 
(computed recursively) with at most $i-1$ levels rooted at~$v$. We prove the 
lemma by induction on $i$.

First consider the length property of $\bar T$. For $i=1$, it
trivially holds by the direct application of the FPTAS 
(line~\ref{linecompl1}). For $i\geq 2$, we can bound the length of $\bar P$ by 
$(1+\varepsilon)^j< \ell(\bar P)\leq (1+\varepsilon)^{j+1}$. By 
line~\ref{linereduce}, the permissible length for a connection from $v$ to some 
node $u$ in $\bar T\setminus\bar P$ is bounded by $d'(u)\leq 
d(u)-(1+\varepsilon)^{j}$. By induction, we will violate this bound by a 
factor of at most $(1+\varepsilon)$, i.e., the length of a connection between $r$ and $u$ in $\bar T$ 
will be at most $(1+\varepsilon)^{j+1} + (1+\varepsilon)(d(u) - 
(1+\varepsilon)^j)=(1+\varepsilon)d(u)$.

Now, consider the cost property. It holds for $i=2$.
Assume $i\geq 3$ and that the claim holds for all level restrictions less 
than~$i$. 
Let $v$ denote a \emph{level-child} of $r$ with respect to $T^*$, i.e., all inner nodes of the path $P_{j,v}$ between $r$ and $v$ in 
$T^*$ are of degree 2. The subtree $T_v\subset T^*$ rooted at $v$ has (at most) $i-1$ 
levels. (By augmenting $G$ with sufficient 0-cost 0-length edges, we can assume that $T^*_v$ 
has precisely $i-1$ levels.) Let $c_{j,v}$ and $\ell_v\leq(1+\varepsilon)^j$ denote 
the cost and length of $P_{j,v}$, respectively. Let $C_v$ denote the cost of $T_v$ 
and $k_v$ the number of terminals in $T_v$. In the following, consider the 
node $v^*$, level-child of $r$ in $T^*$, with minimal $\varrho_{v^*}:= 
(c_{j,v^*}+C_{v^*})/k_{v^*}<\varrho^*$.

At some point at level $i$, our algorithm will also consider node $v^*$ and number 
$k_{v^*}$. The computed $r$-$v^*$ path may be up to 
$(1+\varepsilon)\ell_{v^*}\leq(1+\varepsilon)^{j+1}$ long. We investigate the 
behavior of \textsc{ShallowLight}$(i-1,{v^*},R,d',k_{v^*})$. It returns an 
$(i-1)$-level tree $S$ that is, again, iteratively constructed. Let $S'$ be 
the tree incorporated into $S$ by the algorithm such that the current $S$ now 
contains at least $k_{v^*}/(i-1)$ terminals for the first time. Let $S_0,S_1$ be 
the solution trees before and after adding $S'$, respectively. Furthermore, let 
$s_0,s_1$ be the number of $\bar R$-nodes covered by $S_0,S_1$, respectively. 
Observe that $s_1\geq k_{v^*}/(i-1)$.

Consider the nodes not covered before $S'$: $|T_{v^*}\cap\bar R|\geq k_{v^*} 
- s_0 = k_{v^*} - k_{v^*}/(i-1)=\frac{i-2}{i-1}k_{v^*}$. Since we can cover all 
these nodes at cost at most $C_{v^*}$, we have an upper bound of 
$\frac{i-1}{i-2}C_{v^*}/k_{v^*}$ on the relative cost for the uncovered 
terminals. By our induction hypothesis, we know that we will hence find a 
solution---violating the length restrictions by at most a factor of 
$(1+\varepsilon)$---with relative cost at most 
$(i-2)\frac{i-1}{i-2}C_{v^*}/k_{v^*}$ for $S'$. This upper bound naturally holds 
for each subtree that is incorporated into $S$ before $S'$. Consequently, the 
relative cost of $S_1$ is also at most $(i-1)C_{v^*}/k_{v^*}$.

Now, observe that our algorithm will not only compute 
\textsc{ShallowLight}$(i-1,{v^*},R,d',k_{v^*})$ but also 
\textsc{ShallowLight}$(i-1,{v^*},R,d',s_1)$. Observe the equally modified length 
restrictions~$d'$. In the latter case, the algorithm will stop after adding $S'$ 
to $S$, returning this $S$ as its $(i-1)$-level solution tree of relative cost 
$\varrho(S)\leq(i-2)C_{v^*}/k_{v^*}$. On level $i$, this $S$ will be joined with 
the computed path $\bar P$ of cost at most that of $P_{j,v^*}$ (with corresponding~$j$) 
and violating the length constraints by at most $(1+\varepsilon)$ as discussed 
above. Together, they form a tree $T'$ with $\varrho(T')=\varrho(S)+c_{j,v^*}/s_1 \leq 
(i-2)C_{v^*}/k_{v^*} + c_{j,v^*}/(k_{v^*}/(i-1)) \leq 
(i-1)(c_{j,v^*}+C_{v^*})/k_{v^*}=(i-1)\varrho_{v^*} = (i-1)\varrho^*$.
\end{proof}

We are now able to prove the approximation result for directed shallow-light Steiner 
trees.

\begin{theorem}[Revisited]
 The above algorithm is a bi-criteria 
$(1+\varepsilon_1,O(|R|^{\varepsilon_2}))$-approximation for directed shallow-light Steiner trees: for 
arbitrary small $\varepsilon_1,\varepsilon_2>0$, it gives a solution at most 
$O(|R|^{\varepsilon_2})$ times more expensive than the optimum, while violating the 
length constraints by a factor of at most $(1+{\varepsilon_1})$. For fixed 
$\varepsilon_2$, its runtime is polynomial in the input size and~$\varepsilon_1$.
\end{theorem}
\begin{proof}
Lemma~\ref{partialapprox} shows that each chosen $T_\mathrm{best}$ on level $i$ 
has a relative cost of at most $(i-1)$ the relative-cost-optimum $i$-level tree 
w.r.t.\ $\bar R, \bar k$. By Lemma~\ref{withoutlevels}, the latter approximates 
the optimum tree without level restrictions. So, overall, each $T_\mathrm{best}$ 
is a $(i-1)2i(\bar k/2)^{1/i}$-partial approximation for $k$-DSLST. By 
Lemma~\ref{partialtofull}, this gives a $g(k)$-approximation for $k$-DSLST with 
$$g(k) = \int_0^k \left((i-1)2i(\bar x/2)^{1/i}/x\right)dx = 
\frac{2i^2(i-1)}{2^{1/i}}k^{1/i}.$$ We hence have an 
$O(|R|^{\varepsilon_2})$-approximation for directed shallow-light Steiner trees (=$|R|$-DSLST)---w.r.t.\ violating 
the length bounds by at most a factor of $(1+\varepsilon_1)$---by choosing a 
suitable $i$ inversely correlated to~$\varepsilon_2$.

Consider the running time of our algorithm. {\sc MinCostPath} is an FPTAS 
with running time $O(mn/\varepsilon_1)$~\cite{ergun-etal:fptaspath}. Consider any 
call to {\sc ShallowLight} w.r.t.\ some $i,k$. For $i=1$, it requires 
$O(|R|nm/\varepsilon_1)$ time. Otherwise, we may add $O(k)$ different trees
$T_\mathrm{best}$ and the block of lines \ref{searchbegin}--\ref{searchend} is 
repeated $O(nk^2\log\ell(E))$ times. Overall, any run of the procedure 
(disregarding recursive calls) requires $O(n^2mk^2\log\ell(E)/\varepsilon_1)$ 
time. For overall $i$ levels, there are $O(n^{i-1}k^{2i-2})$ recursive 
invocations, inducing an overall runtime of 
$O(n^{i+1}mk^{2i}\log\ell(E)/\varepsilon_1)$. Clearly, $\log\ell(E)$, the 
logarithm of the sum of all edge lengths, is polynomially bounded by the input 
size, and, by choice of $i$ above, $i$ is directly correlated to (and only dependent on) $1/\varepsilon_2$.
\end{proof}

\section{Conclusions: Light-Weight Directed Spanners}\label{sec:lwds}

We conclude with sketching another application of our
shallow-light Steiner tree result. We obtain a 
bi-criteria approximation algorithm for light-weight directed 
$\alpha$-spanners (cf.\ Definition~\ref{def:alphaspan}).  To the best of
our knowledge no non-trivial result is known for this problem.

We employ a two-stage approach
similar to the one used for directed sparse
spanners~\cite{dinitz11:directed-spanners,berman-etal:apx-spanners-dsf}
and for our network design problem in
Section~\ref{sec:network-design-bounded}.  Thin and thick pairs are
defined analogously to
Section~\ref{sec:network-design-bounded}.  Thin pairs can be settled
as in \cite{berman-etal:apx-spanners-dsf} as only the linearity of the
objective function is used there.  For settling thick pairs, a set of
$\Theta(\sqrt{n}\log n)$ many nodes is sampled.  In the case of
sparse spanners \cite{berman-etal:apx-spanners-dsf} it is sufficient
to compute a shortest path in-tree and
a shortest path out-tree for each of these sampled nodes, and take the union of these trees.  Since
each of these trees has at most $n-1$ edges, which is clearly a lower
bound on $\opt$, the total cost for this stage is $\tilde
O(\sqrt{n}\cdot\opt)$.  It is shown that this procedure settles all
thick pairs with high probability. In the case of light-weight
spanners we compute a directed shallow-light spanning
tree for each sampled node.  More precisely, let $u$ be the sampled node.  We compute a
shallow-light spanning tree $T$ rooted at $u$ such that for each node
$v\in V$ its distance $\bar\ell_T(u,v)$ is at most
$\alpha\cdot\bar\ell_G(u,v)$.  Since the optimum solution to the spanner
problem ensures the existence of a feasible solution to this problem, we can compute
such a tree of cost at most $O(n^\varepsilon\opt)$ using Theorem~\ref{thm:slst}.  Analogously, we can
compute an in-tree with root $u$ and the respective distance bounds.
The total cost of the union of all such spanning trees is
$O(n^{1/2+\varepsilon}\opt)$.  

Unfortunately, the resulting solution is
not necessarily feasible since the stretch factor~$\alpha$ may be
violated.  We can still argue that the solution gives a bi-criteria
approximation with bounded stretch factor.  To see this, consider a
thick pair $(u,v)$ and assume that we sample a node $z$ such that
there is a $u$-$v$ path visiting $z$ of length at most $\alpha\cdot
\bar\ell_G(u,v)$.  Hence $\bar\ell_G(u,z)+\bar\ell_G(z,v)\leq\alpha\bar\ell_G(u,v)$.
Using the paths provided by the shallow-light in-tree and the
shallow-light out-tree computed by our algorithm we can find a path of
length at most
$(\alpha+\varepsilon)\alpha\bar\ell_G(u,z)+(\alpha+\varepsilon)\alpha\bar\ell_G(z,v)\leq(\alpha+\varepsilon)\alpha\bar\ell_G(u,v)$
in our output graph.  We have:
\begin{theorem}[Revisited]
  The above algorithm is a bi-criteria
  $(\alpha+\varepsilon,O(n^{1/2+\varepsilon}))$-approximation for
  light-weight directed $\alpha$-spanners.  The running time depends
  on $n$ and $\varepsilon$ and is polynomial in $n$ for any fixed
  $\varepsilon>0$.
\end{theorem}

\bibliography{spanners}

\end{document}